\documentclass[12pt]{article}

\usepackage[style=alphabetic, backend=biber, maxbibnames=99, sorting=nyt, doi=false, url=false, isbn=false]{biblatex}

\DeclareSourcemap{
	\maps[datatype=bibtex]{
		\map[overwrite=true]{
			\step[fieldsource=fjournal]
			\step[fieldset=journal, origfieldval]
			\step[fieldset=language, null]
			\step[fieldset=issn, null]
			\step[fieldset=month, null]
        }
        \map[overwrite=true]{
            \pertype{article}
            \pertype{inproceedings}
            \pertype{unpublished}
			\step[fieldset=publisher, null]
        }
	}
}

\usepackage[letterpaper, margin=1in]{geometry}

\usepackage{amsmath, amssymb, mathtools, amsthm}
\usepackage{graphicx}
\usepackage[colorlinks=true, allcolors=blue]{hyperref}
\usepackage{algorithm, algpseudocode}

\newtheorem{theorem}{Theorem}
\newtheorem{lemma}[theorem]{Lemma}
\newtheorem{question}{Question}

\theoremstyle{definition}
\newtheorem{definition}{Definition}

\newcommand*{\bigO}{\mathop{\mathrm{O}}}
\newcommand*{\mul}{\mathop{\mathsf{M}}}
\newcommand*{\bideg}{\mathop{\mathrm{bideg}}}
\newcommand*{\extract}{\mathop{\mathcal{F}}\nolimits}
\newcommand*{\eval}[1]{\left.{#1}\right\rvert}
\newcommand*{\Comp}{\mathop{\mathsf{Comp}}}
\newcommand*{\PowProj}{\mathop{\mathsf{PowProj}}}

\newcommand*{\ring}{\mathbb{A}}

\title{Power Series Composition in Near-Linear Time}
\author{
    Yasunori Kinoshita
    \thanks{Tokyo Institute of Technology, Japan.} \and
    Baitian Li
    \thanks{Institute for Interdisciplinary Information Sciences, Tsinghua University.
    Email: \texttt{lbt21@mails.tsinghua.edu.cn}.}
}
\date{}

\begin{document}
\maketitle
\begin{abstract}
    We present an algebraic algorithm that computes the composition of two power series in
    softly linear time complexity. The previous best algorithms
    are $\bigO(n^{1+o(1)})$ non-algebraic algorithm by Kedlaya and Umans (FOCS 2008) and an $\bigO(n^{1.43})$ algebraic algorithm
    by Neiger, Salvy, Schost and Villard (JACM 2023).

    Our algorithm builds upon the recent Graeffe iteration
    approach to manipulate rational power series introduced by Bostan and Mori (SOSA 2021).
\end{abstract}

\section{Introduction}

Let $\ring$ be a commutative ring and let $f(x), g(x)$ be polynomials in $\ring[x]$ of degrees less
than $m$ and $n$, respectively. The problem of \emph{power series composition} is to compute
the coefficients of $f(g(x)) \bmod x^n$.
The terminology stems from the idea that $g(x)$ can be seen as a truncated formal power series.
This is a fundamental problem in computer algebra~\cite[Section ~4.7]{Knuth98TAOCP2},
and has applications in various areas such as combinatorics \cite{PSS12Comb}
and cryptography~\cite{BMS08Isog}.

A textbook algorithm for power series composition is due to Brent and Kung \cite{brent1978fast},
which computes the composition in $\bigO(\mul(n) (n \log n)^{1/2})$ time complexity.
Here $\mul(n)$ denotes the complexity of computing the product of two polynomials of degree less than $n$.

The current best known algorithm for power series composition is due to Kedlaya and Umans \cite{KU08PolyComp, KU11PolyFact},
which computes the composition in $(n \log q)^{1 + o(1)}$ bit operations, where $\ring$ is the finite field $\mathbb{F}_q$.
Van der Hoeven and Lecerf \cite{vdHL20Revisit} gave a detailed analysis of the subpolynomial term,
and showed that the algorithm has a time complexity of $\tilde \bigO(n 2^{\bigO(\sqrt{\log n \log \log n})} \log q)$.

In this paper, we present a simple algorithm that reduces the complexity of power series composition to near-linear time.
Furthermore, our algorithm works over arbitrary commutative rings.

\begin{theorem} \label{thm:main}
    Given polynomials $f(x), g(x) \in \ring[x]$ with degree less than $m$ and $n$, respectively,
    the power series composition
    $f(g(x)) \bmod x^n$ can be computed in $\bigO(\mul(n)\log m + \mul(m))$ arithmetic operations.
\end{theorem}

We remark that our algorithm also has consequences for other computation models measured by bit complexity,
such as boolean circuits and multitape Turing machines.

\paragraph{Technical Overview.}

One of our key ingredients is Graeffe's root squaring method, which was first time introduced by Sch\"onhage \cite{Sch00Recip}
in a purely algebraic setting to compute reciprocals of power series. Recently, Bostan and Mori \cite{bostan2021simple} applied
the Graeffe iteration to compute the coefficient of $x^N$ of the series expansion of a rational power series
$P(x)/Q(x)$, achieving time complexity $\bigO(\mul(n) \log N)$ where $n$ is the degree bound of $P(x)$ and $Q(x)$.

In a nutshell, their algorithm works as follows. Considering multiplying $Q(-x)$ on both the numerator and the denominator,
one obtains $[x^N]\footnote{We use $[x^N]$ to denote coefficient extraction, i.e., for a power series $F(x)$, $[x^N]F(x)$ is the coefficient of $x^N$ in $F(x)$.} \frac{P(x)}{Q(x)} = [x^N] \frac{P(x)Q(-x)}{Q(x)Q(-x)}$. By the symmetry of $Q(x)Q(-x)$, one can
rewrite the expression to $[x^N] \frac{U(x)}{V(x^2)}$ where $U(x) = P(x)Q(-x)$ and $V(x^2) = Q(x)Q(-x)$.
Furthermore, one can split $U(x)$ into even and odd parts as $U(x) = U_e(x^2) + xU_o(x^2)$, and
reduce the problem into $[x^{N/2}] \frac{U_e(x)}{V(x)}$ or $[x^{(N-1)/2}] \frac{U_o(x)}{V(x)}$ depending on the parity of $N$.
Since the iteration reduces $N$ by half each time, the algorithm has time complexity $\bigO(\mul(n) \log N)$
by using polynomial multiplication.

Despite their algorithm has the same time complexity as the classical algorithm due to Fiduccia~\cite{Fid85Rec},
their algorithm is based on a totally different idea, and it is worth to get a closer look at their algorithm
when $N \leq n$. In that case, the terms greater than $N$ are irrelevant, hence the size of the problem can be reduced to $N$,
compared with $n$. Since the size of the problem is reduced by half in each iteration,
the total time complexity is $\bigO(\mul(N))$.

We extend the Graeffe iteration approach to a special kind of bivariate
rational power series $P(x) / Q(x, y)$, where $Q(x, y) = 1/(1 - y g(x))$.
The goal is still to compute the coefficient of $x^n$ of the series expansion of $P(x) / Q(x, y)$,
but in this case, the output is a polynomial in $y$, instead of a single $\ring$-coefficient.

At the beginning of our algorithm, both $P$ and $Q$ have a degree of $n$ with respect to $x$, and a degree at most $1$ with respect to $y$.
In each iteration, we halve the degree with respect to $x$, and double the degree with respect to $y$,
therefore the size of the problem stays at $n$. Since there are $\log n$ iterations, the total time
complexity is $\bigO(\mul(n) \log n)$.

We observe that the above algorithm, actually already solved the power projection problem,
which is the transposed problem of power series composition.
By the transposition principle --- a general method for transforming algebraic algorithms into algorithms for the transposed problem
with nearly identical complexity, we conclude the algorithm for power series composition.
For simplicity, we unpack the transposition principle and give an independent
explanation of the algorithm for power series composition in our presentation.

We first discuss the details of the Graeffe iteration algorithm for power projection in Section \ref{sec:PowProj},
as well as the analysis of correctness and running time. Then we explicitly present the algorithm for the power
series composition in Section \ref{sec:Comp}.

\subsection{Related Work}

\paragraph{Special Cases.}
Many previous works have focused on special cases where $f$ and/or $g$ has a specific form.
Kung~\cite{kung1974computing} applied Newton iteration to compute the reciprocal of a power series,
i.e., $f = (1+x)^{-1}$, in $\bigO(\mul(n))$ operations,
Brent \cite{Brent76Ele} further extended the algorithm for $f$ being elementary functions.
Extensive work has been done to further optimize the constant factor of the algorithms
computing reciprocals \cite{SGV94FastAlg,Ber04Newton,Sch00Recip,Har11Faster,Ser10SOTA},
square roots \cite{Ber04Newton,Har11Faster,Ser10SOTA},
exponentials \cite{Ber04Newton,vdH10NewtonFFT,BS09ExpFast,Har09FasterExp,Ser10SOTA,Ser12note},
and other elementary functions.

Gerhard \cite{Gerhard00CompExp} gave an $\bigO(\mul(n)\log n)$
algorithm for the case when $g(x) = e^x-1$ and $g(x) = \log(1+x)$,
which can be interpreted as the conversion between the falling factorial basis and the monomial basis
of polynomials.
Brent and Kung \cite{brent1978fast} showed that when $g(x)$ is a constant degree polynomial, the composition can be computed
in $\bigO(\mul(n)\log n)$ time. This result is generalized by van der Hoeven~\cite{vdH02Relax} to the case where $g(x)$
is an algebraic power series.
Bostan, Salvy and Schost~\cite{BSS08Comp} isolated a large class of $g$ such that the composition problem
can be computed in $\bigO(\mul(n) \log n)$ or $\bigO(\mul(n))$ operations, by composing several
basic subroutines and the usage of the transposition principle \cite{BLG03Tellegen}.

Bernstein \cite{Bernstein98SmallChar} studied the composition of a power series in the case where $\ring$ has a small characteristic.
When $\ring$ has positive characteristic $q$,
his algorithm computes the composition of a power series in $\bigO((q/\log q)\mul(n) \log n)$ operations. Ritzmann \cite{Ritzmann1986Numeric} studied the numerical algorithm of power series composition, where $\ring$ is the field of complex numbers.

\paragraph{Modular Composition.}
Besides Brent and Kung's $\bigO(\mul(n) (n\log n)^{1/2})$-time algorithm, all known approaches
toward the general power series composition, actually solve a generalized problem called the modular composition problem.
This problem is, given three polynomials $f(x), g(x), h(x) \in \ring[x]$ where $h$ is monic,
to compute $f(g) \bmod h$
where the $\bmod$ operation takes the remainder of the Euclidean division. When restricted to $h(x) = x^n$,
the modular composition problem becomes the power series composition problem.

In Brent and Kung's paper \cite{brent1978fast} on power series composition, they also presented an algorithm that computes
modular composition in $\bigO(n^{(\omega+1)/2} + \mul(n) (n \log n)^{1/2})$ operations,
where $\omega$ is a feasible exponent of matrix multiplication. Huang and Pan \cite{HP98RectMult} improved the bound to
$\bigO(n^{\omega_2/ 2} + \mul(n) (n \log n)^{1/2})$ by using the fast rectangular matrix multiplication,
where $\omega_2 \leq \omega+1$ is a feasible exponent of rectangular matrix multiplication of size $n \times n^2 \times n$.
Even assuming one can attain the obvious lower bounds $\omega \geq 2$ and $\omega_2 \geq 3$,
their algorithms cannot yield a bound better than $\bigO(n^{1.5})$.

Based on this, B\"urgisser, Clausen and Shokrollahi \cite[open problem 2.4]{BCS10AlgComp} 
and von zur Gathen and Gerhard \cite[research problem 12.19]{ModernComputerAlgebra} asked the following question:

\begin{question}\label{question:modcomp}
    Is there an algorithm that computes modular composition better than Brent and Kung's
    approach, or even better than $\bigO(n^{1.5})$?
\end{question}

Their question is partially answered by Kedlaya and Umans,
and Neiger, Salvy, Schost and Villard.

Kedlaya and Umans \cite{KU08PolyComp, KU11PolyFact} presented two algorithms for computing modular composition, where $\ring$ is the finite field $\mathbb{F}_q$. The first algorithm performs $(n \log q)^{1 + o(1)}$ bit operations, while the second one performs $p^{1+o(1)} n^{1+o(1)}$ algebraic operations, where $p$ is the characteristic of the field. However, the first algorithm involves non-algebraic operations, such as lifting from the prime field to the integers, and the second algorithm is efficient only for fields with small characteristic. As a result, their approach is unlikely to extend to the case of general rings.

A recent breakthrough by Neiger, Salvy, Schost and Villard \cite{NSSV23PolyComp} presented a Las Vegas algorithm
that computes modular composition in $\tilde\bigO(n^\kappa)$\footnote{We use $\tilde \bigO$ to suppress log factors, i.e., $T(n) = \tilde\bigO(f(n))$ denotes $T(n) = \bigO(f(n) \log^k f(n))$ for some $k$.} arithmetic operations, assuming $\ring$ is a field
and some mild technical assumptions. Here $\kappa$ is determined by
\[ \kappa = 1 + \frac 1{\frac 1{\omega - 1} + \frac{2}{\omega_2 - 2}}. \]
By the best known upper bound $\omega \leq 2.371552$ and $\omega_2 \leq 3.250385$ given by Williams,
Xu, Xu and Zhou \cite{WXXZ24MatrixMult}, it follows that $\kappa \leq 1.43$, which breaks the $n^{1.5}$ barrier.
However, by the obvious lower bound $\omega \geq 2$ and $\omega_2 \geq 3$, their algorithm cannot yield a bound better than
$\bigO(n^{4/3})$.

Our results can be seen as another partial answer to Question \ref{question:modcomp}.

\section{Preliminaries}

\subsection{Notation}

In this paper, $\ring$ denotes an effective ring, i.e., a commutative ring with unity,
and one can perform the basic ring operations $(+, -, \times)$ with unit cost in $\ring$.
The polynomial ring and the formal power series ring over $\ring$ are denoted by $\ring[x]$ and $\ring[[x]]$, respectively.
We use $\ring[x]_{<n}$ to denote the set of polynomials of degree less than $n$.
For any polynomial or power series $f(x) = \sum_i f_i x^i$, $[x^i]f(x)$ denotes the coefficient $f_i$ of $x^i$,
and $f(x) \bmod x^n$ denotes the truncation $\sum_{i=0}^{n-1} f_i x^i$.

For a bivariate polynomial $f(x,y) \in \ring[x,y]$, $\deg_x f(x, y)$ and $\deg_y f(x, y)$ denote its degree with
respect to $x$ and $y$, respectively.
The bidegree of $f(x,y)$ is the pair $\bideg f(x, y) = (\deg_x f(x, y), \deg_y f(x, y))$.
Inequalities between bidegrees are component-wise. For example, $\bideg f(x,y) \leq (n,m)$
means $\deg_x f(x,y) \leq n$ and $\deg_y f(x,y) \leq m$.

\subsection{Computation Model}

Informally, we measure the complexity of an algorithm by counting the number of arithmetic operations in the ring $\ring$.
The underlying computation model is the straight-line program. Readers are referred to
the textbook of B\"urgisser, Clausen and Shokrollahi
\cite[Sec.~4.1]{BCS10AlgComp} for a detailed explanation.

\subsection{Polynomial Multiplication}

Our algorithm calls polynomial multiplication as a crucial subroutine.
Over different rings, the time complexity of our algorithm may vary due
to the different polynomial multiplication algorithms.

For an effective ring $\ring$, we call a function $\mul_\ring(n)$ a \emph{multiplication time} for $\ring[x]$
if polynomials in $\ring[x]$ of degree less than $n$ can be multiplied in $\bigO(\mul_\ring(n))$ operations.
When the ring $\ring$ is clear from the context, we simply write $\mul(n)$.

Throughout this paper, we assume a regularity condition on $\mul(n)$,
that for some constants $\alpha, \beta \in (0,1)$, the function $\mul$ satisfies the condition
$\mul(\lceil \alpha n \rceil) \leq \beta \mul(n)$ for all sufficiently large $n$.
All the examples of $\mul$ presented in this paper satisfy this condition.

The fast Fourier transform (FFT) \cite{CT65FFT} gives the bound
when $\ring$ is a field possessing roots of unity of sufficiently high order;
for example, the complex numbers $\mathbb{C}$ satisfies $\mul_{\mathbb C}(n) = \bigO(n\log n)$.
For any $\mathbb F_p$-algebra $\ring$ where $p$ is a prime, Harvey and van der Hoeven \cite{HvdH19Uncond}
proved that $\mul_\ring(n) = \bigO(n \log n \, 4^{\log ^* n})$\footnote{$\log ^*$ is the iterated logarithm.}.
For general rings, Cantor and Kaltofen \cite{cantor1987fast} showed a uniform bound $\mul_\ring(n) = \bigO(n\log n \log \log n)$.
This implies that our algorithm has a uniform bound $\bigO(n\log n \log m \log \log n + m\log m\log \log m)$.

We may also obtain complexity bounds for various computation models,
depending on the fast arithmetic of $\ring$ and the polynomial multiplication algorithm on the corresponding model.
For example, for the bit complexity of boolean circuits and multitape Turing machines,
since the polynomial multiplication over $\mathbb F_p$ can be done in
$\mul_{\mathbb F_p}(n) = \bigO((n\log p)\log (n\log p) 4^{\max(\log^* p, \log^* n)})$ bit operations \cite{HvdH19Uncond},
substituting such a bound into our algorithm yields a bound for the bit complexity of the composition problem.

\subsection{Operations on Power Series and Polynomials}

We need the following operations on power series and polynomials that are used in our algorithm.

\begin{theorem}[{\cite{kung1974computing}}] \label{thm:recip}
    Let $\ring$ be an effective ring and
    $f(x) \in \ring[x]_{<n}$ satisfies $f(0) = 1$, the truncation of
    the reciprocal $f(x)^{-1} \bmod x^n$ can be computed in $\bigO(\mul_\ring(n))$ operations.
\end{theorem}

Using Kronecker's substitution, one can reduce the problem of computing the multiplication of bivariate
polynomials to univariate cases.
\begin{theorem}[{\cite[Corollary 8.28]{ModernComputerAlgebra}}] \label{thm:kronecker}
    Let $f(x,y), g(x,y) \in \ring[x,y]$ with bidegree bound $(n,m)$,
    their product $f(x,y)g(x,y)$ can be computed in $\bigO(\mul_\ring(nm))$ operations.
\end{theorem}

We also require the following lemma, which bounds the total time complexity of polynomial multiplications with exponentially decreasing degrees.
\begin{lemma}[{\cite[Lemma 1.1]{brent1978fast}}] \label{lem:geomseries}
	$\sum_{i=0}^{\lfloor \log_2 n \rfloor} \mul(\lceil n / 2^i \rceil) = \bigO(\mul(n))$.
\end{lemma}

\subsection{Transposition Principle}

The transposition principle, a.k.a.~Tellegen's principle, is an algorithmic theorem
that allows one to produce a new algorithm of the \emph{transposed problem} of a given linear problem.

\begin{theorem}[{\cite[Thm.~13.20]{BCS10AlgComp}}]
    Let $M$ be an $r \times s$ matrix. Suppose that there exists a linear
    straight-line program of length $L$ that computes the matrix-vector product $b \mapsto M b$,
    and $M$ has $z_0$ zero rows and $z_1$ zero columns. Then the transposed problem
    $c\mapsto M^{\mathsf T}c$ can be solved by a linear straight-line program of length
    $L - s + r - z_1 + z_0$.
\end{theorem}

Roughly speaking, the transposition principle allows one to compute the transposed problem
with the same time complexity as the original problem, provided that the algorithm is $\ring$-linear,
in the sense that only linear operations in the coefficients of $b$ are performed.

The transposed problem of power series composition is the power projection problem, which is defined as follows.

\begin{definition}[Power Projection]
    Given a positive integer $m$, a polynomial $g(x)$ of degree less than $n$, and a linear form
    $w \colon \ring[x]_{< n} \to \ring$, the power projection problem is to compute
    $f_i = w(g(x)^i \bmod x^n)$
    for all $i = 0,\dots,m-1$. The linear form $w$ is identified by its values on the monomials $x^i$,
    i.e., $w(\sum_{i = 0}^{n-1} a_i x^i) = \sum_{i=0}^{n-1} w_i a_i$.
\end{definition}

For a given polynomial $g(x)$, consider the $n\times m$ matrix $M$ given by
$M_{ij} = [x^i] g(x)^j$. The power series composition problem is equivalent to computing
$M f$, where $f$ is the column vector of coefficients of $f(x)$,
and the power projection problem is equivalent to computing $M^{\mathsf T} b$,
where $b$ is the column vector of coefficients of the linear form $w$.
It is important to note that the power projection and composition problems are linear only when the polynomial $g$ is \emph{fixed}, as the output of these problems is nonlinear with respect to the coefficients of $g$.

\section{Main Results}

In Section~\ref{sec:PowProj}, we present an efficient $\ring$-linear algorithm for power projection.
By applying the transposition principle to this algorithm, we obtain an algorithm for power series composition.
In Section~\ref{sec:Comp}, we describe the derivation of the power series composition algorithm without relying on the transposition principle, thereby providing a direct proof of Theorem~\ref{thm:main}.

\subsection{Algorithm for Power Projection} \label{sec:PowProj}

Our objective is to compute $f_i = w(g(x)^i \bmod x^n)$ for all $i=0, \dots, m-1$.
We define a polynomial $w^{\mathrm{R}}$ as $w^{\mathrm{R}}(x) = \sum_{j=0}^{n-1} w_{n-1-j} x^j$.
Let $P(x, y) \coloneqq w^{\mathrm{R}}(x)$ and $Q(x,y) \coloneqq 1 - y g(x)$.
Then, the generating function $f(y) \coloneqq \sum_{i=0}^{m-1} f_i y^i$ is rearranged as:
\begin{equation*}
    \begin{split}
        f(y)
        &= \sum_{i=0}^{m-1} y^i \sum_{j=0}^{n-1} w_j [x^j] g(x)^i \\
        &= \sum_{i=0}^{m-1} y^i [x^{n-1}] ( w^{\mathrm{R}}(x) g(x)^i ) \\
        &= \left(\sum_{i=0}^{\infty} y^i [x^{n-1}] ( w^{\mathrm{R}}(x) g(x)^i )\right) \bmod y^m \\
        &= [x^{n-1}]\left( w^{\mathrm{R}}(x) \sum_{i=0}^{\infty} y^i g(x)^i \right) \bmod y^m \\
        &= [x^{n-1}]\frac{w^{\mathrm{R}}(x)}{1 - y g(x)} \bmod y^m \\
        &= [x^{n-1}]\frac{P(x, y)}{Q(x,y)} \bmod y^m
        \eqqcolon u.
    \end{split}
\end{equation*}
If $n - 1 = 0$, then $u = P(0,y) / Q(0,y)$.
If $n - 1 \geq 1$, we multiply $Q(-x,y)$ by the numerator and the denominator, resulting in
\begin{equation*}
    u = [x^{n-1}]{\frac{P(x,y)Q(-x,y) \bmod x^n \bmod y^m}{Q(x,y)Q(-x,y) \bmod x^n \bmod y^m}} \bmod y^m.
\end{equation*}
Let $A(x,y) \coloneqq Q(x,y)Q(-x,y) \bmod x^n \bmod y^m$.
$A(x,y)$ satisfies the equation $A(x,y) = A(-x,y)$, which implies $[x^i] A(x,y) = 0$ for all odd $i$; thus, there exists a unique polynomial $V(x,y)$ such that $V(x^2,y) = A(x,y)$.
Now, we have
\begin{equation*}
    u = [x^{n-1}]\frac{U(x,y)}{V(x^2,y)} \bmod y^m
\end{equation*}
for $U(x,y) \coloneqq P(x,y)Q(-x,y) \bmod x^n \bmod y^m$.
Since $1 / V(x^2, y)$ is an even formal power series with respect to $x$, i.e., $[x^i](1/V(x,y))=0$ for all odd $i$, we can ignore the odd (or even) part of $U(x,y)$ if $n-1$ is even (or odd).
Let $U_e(x,y)$ and $U_o(x, y)$ be the unique polynomials such that $U(x,y) = U_e(x^2,y) + xU_o(x^2,y)$.
Then, we have
\begin{equation*}
    \begin{split}
        u
        &= \begin{cases}
            \displaystyle [x^{n-1}]\frac{U_e(x^2,y)}{V(x^2,y)} \quad &(n-1 \text{ is even}) \\
            \displaystyle [x^{n-1}]\frac{xU_o(x^2,y)}{V(x^2,y)} \quad &(n-1 \text{ is odd})
        \end{cases} \\
        &= \begin{cases}
            \displaystyle [x^{\lceil n/2 \rceil - 1}]\frac{U_e(x,y)}{V(x,y)} \quad &(n-1 \text{ is even}) \\
            \displaystyle [x^{\lceil n/2 \rceil - 1}]\frac{U_o(x,y)}{V(x,y)} \quad &(n-1 \text{ is odd}).
        \end{cases}
    \end{split}
\end{equation*}
Now the problem is reduced to a case where $n$ is halved.
We repeat this reduction until $n - 1$ becomes $0$.
The resulting algorithm is summarized in Algorithm \ref{alg:PowProj}.

\begin{algorithm}
    \caption{
        ($\PowProj$)\\
        \textbf{Input}: $n$, $m$, $P(x,y), Q(x,y) \in \ring[x,y]$ \\
        \textbf{Output}: $[x^{n-1}] (P(x,y) / Q(x,y)) \bmod y^m$
    }
    \label{alg:PowProj}
    \begin{algorithmic}[1]
        \Require $[x^0 y^0] Q(x,y) = 1$
        \While {$n > 1$}
            \State $U(x,y) \gets P(x,y)Q(-x,y) \bmod x^n \bmod y^m$
            \If {$n-1$ is even}
                \State $P(x,y) \gets \sum_{i=0}^{\lceil n/2 \rceil - 1} x^i [s^{2i}] U(s, y)$
            \Else
                \State $P(x,y) \gets \sum_{i=0}^{\lceil n/2 \rceil - 1} x^i [s^{2i+1}] U(s, y)$
            \EndIf
            \State $A(x,y) \gets Q(x,y)Q(-x,y) \bmod x^n \bmod y^m$
            \State $Q(x,y) \gets \sum_{i=0}^{\lceil n/2 \rceil - 1} x^i [s^{2i}] A(s,y)$
            \State $n \gets \lceil n/2 \rceil$
        \EndWhile
        \State \Return $( P(0, y) / Q(0, y) ) \bmod y^m$ \label{line:PowProjRecip}
    \end{algorithmic}
\end{algorithm}

\begin{theorem}
    Given a positive integer $m$, a polynomial $g(x) \in \ring[x]$ with degree less than $n$, and a linear form $w \colon \ring[x]_{<n} \to \ring$, one can compute $w(g(x)^i \bmod x^n) = \sum_{j=0}^{n-1} w_j [x^j] g(x)^i$ for all $i=0,\dots,m-1$ in $\bigO(\mul(n) \log m + \mul(m))$ operations.
\end{theorem}
\begin{proof}
    We analyze the time complexity of $\PowProj(n,m,\sum_{j=0}^{n-1}w_{n-1-j}x^j,1 - yg(x))$.
    Let $T_k$ represent the time complexity required for polynomial multiplications in the $k$-th iteration out of a total of $\lceil \log_2 n \rceil$ iterations in Algorithm~\ref{alg:PowProj}.
    Here, we consider the first iteration as the $0$-th iteration.
    By induction, we observe that both $\bideg P$ and $\bideg Q$ are at most $(\lceil n / 2^k \rceil - 1, 2^k)$ in the $k$-th iteration, which implies $T_k = \bigO(\mul(n))$.
    At the same time, we also have both $\bideg P$ and $\bideg Q$ are at most $(\lceil n / 2^k \rceil - 1, m - 1)$, which implies $T_k = \bigO(\mul(\lceil nm / 2^k \rceil))$.
    By combining these bounds with Lemma~\ref{lem:geomseries}, the total time complexity of all iterations is bounded as
    \begin{equation*}
        \begin{split}
            \sum_{k=0}^{\lceil \log_2 n \rceil - 1} T_k
            &= \bigO \left( \sum_{k=0}^{\lceil \log_2 m \rceil-1} \mul(n) \right) + \bigO \left( \sum_{k=\lceil \log_2 m \rceil}^{\lceil \log_2 n \rceil - 1} \mul(\lceil nm / 2^k \rceil) \right) \\
            &= \bigO(\mul(n) \log m) + \bigO(\mul(n)) \\
            &= \bigO(\mul(n) \log m).
        \end{split}
    \end{equation*}
    The computation in line \ref{line:PowProjRecip} can be performed in $\bigO(\mul(m))$ operations by Theorem~\ref{thm:recip}.
    Other parts of the algorithm have a negligible contribution to the overall time complexity.
    Consequently, the total time complexity is $\bigO(\mul(n) \log m + \mul(m))$.
\end{proof}

\subsection{Algorithm for Composition} \label{sec:Comp}

The algorithm presented in this section is essentially the transposition of the algorithm presented in Section~\ref{sec:PowProj}.
However, to enhance comprehension, we present a derivation that does not rely on the transposition principle.
To this end, we rewrite the objective $f(g(x)) \bmod x^n$.

Let $P(y) \coloneqq y^{m-1} f(y^{-1}), Q(x,y) \coloneqq 1 - yg(x)$.
Then, we observe that
\begin{equation*}
    \begin{split}
        f(g(x)) \bmod x^n
        &= \sum_{i=0}^{m-1} ([y^i]f(y)) g(x)^i \bmod x^n \\
        &= [y^{m-1}] \left( y^{m-1} f(y^{-1}) \sum_{i=0}^{\infty} y^i g(x)^i \right) \bmod x^n\\
        &= [y^{m-1}] \frac{y^{m-1} f(y^{-1})}{1 - yg(x)} \bmod x^n\\
        &= [y^{m-1}] \frac{P(y)}{Q(x, y)} \bmod x^n.
    \end{split}
\end{equation*}
Let us define a operator $\extract$ by
\begin{equation*}
    \extract_{l,r}\left(\sum_{i \geq 0} a_i(x) y^i\right) \coloneqq \sum_{i=l}^{r-1} a_{i}(x) y^{i-l}.
\end{equation*}
Then our objective is to compute $\extract_{m-1,m}(P(y)/Q(x,y)) \bmod x^n$.
For algorithmic convenience, we generalize our aim to computing $\extract_{d,m}(P(y)/Q(x,y)) \bmod x^n$ for a given $d$.
If $n = 1$, then $\extract_{d,m}(P(y)/Q(x,y)) \bmod x^n = \extract_{d,m}(P(y)/Q(0,y))$.
If $n > 1$, similar to the transformation described in Section~\ref{sec:PowProj}, we can derive that
\begin{equation*}
    \begin{split}
        \extract_{d,m}\left(\frac{P(y)}{Q(x,y)}\right) \bmod x^n
        &= \extract_{d,m}\left(\frac{P(y)Q(-x,y)}{Q(x,y)Q(-x,y) \bmod x^n \bmod y^m}\right) \bmod x^n \\
        &= \extract_{d,m}\left(\frac{P(y)}{V(x^2,y)}Q(-x,y)\right) \bmod x^n
        \eqqcolon u,
    \end{split}
\end{equation*}
where $V(x^2,y) = Q(x,y)Q(-x,y) \bmod x^n \bmod y^m$.
Let $e \coloneqq \max \{0, d - \deg_y Q(x,y)\}$.
Then, we can discard terms of order less than $e$ with respect to $y$ among $P(y)/V(x^2,y)$ to compute $u$, because they has no contribution to term with order $d$ or higher with respect to $y$ when multiplied by $Q(-x,y)$.
Therefore, $u$ can be rewritten as
\begin{equation*}
    \begin{split}
        u
        &= \extract_{d,m} \left(\left(\extract_{e,m}\left(\frac{P(y)}{V(x^2,y)}\right) y^e \right) Q(-x,y) \right) \bmod x^n \\
        &= \extract_{d-e,m-e} \left(\extract_{e,m}\left(\frac{P(y)}{V(x^2,y)}\right) Q(-x,y) \right) \bmod x^n \\
        &= \extract_{d-e,m-e} \left(\left( \extract_{e,m}\left(\frac{P(y)}{V(x^2,y)}\right) \bmod x^n \right) Q(-x,y) \right) \bmod x^n \\
        &= \extract_{d-e,m-e} \left(\eval{ \left( \extract_{e,m}\left(\frac{P(y)}{V(z,y)}\right) \bmod z^{\lceil n/2 \rceil} \right) }_{z=x^2} Q(-x,y) \right) \bmod x^n.
    \end{split}
\end{equation*}
Now the problem is reduced to a case where $n$ is halved.
We repeat this reduction until $n$ becomes $1$.
The resulting algorithm is summarized in Algorithm \ref{alg:Comp}.

\begin{algorithm}
    \caption{
        ($\Comp$)\\
        \textbf{Input}: $n$, $d$, $m$, $P(y) \in \ring[y]$, $Q(x,y) \in \ring[x,y]$ \\
        \textbf{Output}: $\extract_{d,m}(P(y) / Q(x,y)) \bmod x^n$
    }
    \label{alg:Comp}
    \begin{algorithmic}[1]
        \Require $[x^0 y^0] Q(x,y) = 1$
        \If {$n = 1$}
            \State $C(y) \gets (P(y) / Q(0,y)) \bmod y^m$ \label{line:CompRecip}
            \State \Return $\sum_{i=d}^{m-1} y^{i-d} [t^i] C(t)$
        \Else
            \State $A(x,y) \gets Q(x,y)Q(-x,y) \bmod x^n \bmod y^m$ \label{line:CompMul1}
            \State $V(x,y) \gets \sum_{i=0}^{\lceil n/2 \rceil - 1} x^i [s^{2i}] A(s,y)$
            \State $e \gets \max \{ 0, d - \deg_y Q(x,y) \}$
            \State $W(x,y) \gets \Comp(\lceil n/2 \rceil,e,m,P(y),V(x,y))$
            \State $B(x,y) \gets W(x^2,y)Q(-x,y)$ \label{line:CompMul2}
            \State \Return $\sum_{i=d-e}^{m-e-1} y^{i-(d-e)} [t^i] B(x,t) \bmod x^n$
        \EndIf
    \end{algorithmic}
\end{algorithm}

\begin{proof}[Proof of Theorem \ref{thm:main}]
    We analyze the time complexity of $\Comp(n,m-1,m,\ y^{m-1} f(y^{-1}),1 - yg(x))$.
    Let $T_k$ represent the time complexity required for line~\ref{line:CompMul1} and line~\ref{line:CompMul2} in the $k$-th invocation out of a total of $\lceil \log_2 n \rceil$ recursive invocations of Algorithm~\ref{alg:Comp}.
    Here, we consider the first invocation as the $0$-th invocation.
    By induction, we observe that $\bideg Q \leq (\lceil n / 2^k \rceil - 1, 2^k)$ and $\bideg W \leq (\lceil n / 2^{k+1} \rceil - 1, 2^{k+1})$ in the $k$-th invocation, which implies $T_k = \bigO(\mul(n))$.
    At the same time, we also have $\bideg Q \leq (\lceil n / 2^k \rceil - 1, m - 1)$ and $\bideg W \leq (\lceil n / 2^{k+1} \rceil - 1, m - 1)$, which implies $T_k = \bigO(\mul(\lceil nm / 2^k \rceil))$.
    By combining these bounds with Lemma~\ref{lem:geomseries}, the total time complexity of all invocations is bounded as
    \begin{equation*}
        \begin{split}
            \sum_{k=0}^{\lceil \log_2 n \rceil - 1} T_k
            &= \bigO \left( \sum_{k=0}^{\lceil \log_2 m \rceil-1} \mul(n) \right) + \bigO \left( \sum_{k=\lceil \log_2 m \rceil}^{\lceil \log_2 n \rceil - 1} \mul(\lceil nm / 2^k \rceil) \right) \\
            &= \bigO(\mul(n) \log m) + \bigO(\mul(n)) \\
            &= \bigO(\mul(n) \log m).
        \end{split}
    \end{equation*}
    The computation in line \ref{line:CompRecip} can be performed in $\bigO(\mul(m))$ operations by Theorem~\ref{thm:recip}.
    Other parts of the algorithm have a negligible contribution to the overall time complexity.
    Consequently, the total time complexity is $\bigO(\mul(n) \log m + \mul(m))$.
\end{proof}

\section*{Acknowledgements}

We extend our gratitude to Josh Alman and Hanna Sumita for their review and revisions of our paper. We also appreciate the helpful suggestions provided by the anonymous reviewers.

\printbibliography

\end{document}